\documentclass[conference]{IEEEtran}
\IEEEoverridecommandlockouts
\usepackage[noadjust]{cite}
\usepackage{amsmath,amssymb,amsfonts}
\usepackage{algorithmic}
\usepackage{graphicx}
\usepackage{textcomp}
\usepackage{enumitem}
\usepackage{xcolor}
\usepackage{comment}
\usepackage{mathptmx}
\usepackage{amsthm}
\usepackage{changepage}
\usepackage{mathtools}
\usepackage{balance}
\newtheorem{thm}{Theorem}
\newtheorem{lem}{Lemma}

\newtheorem{cor}{Corollary}

\theoremstyle{definition}
\newtheorem{defn}{Definition}

\def\BibTeX{{\rm B\kern-.05em{\sc i\kern-.025em b}\kern-.08em
    T\kern-.1667em\lower.7ex\hbox{E}\kern-.125emX}}

\usepackage{bbm}
\newcommand{\todo}[1]{{\color{red}#1}}

\usepackage{etoolbox}    
\usepackage{dirtytalk}
\newtoggle{singlecolumn}
\togglefalse{singlecolumn} % Use for double column format -- also change documentclass above 

  \newenvironment{proofachievable}{%
   \proof}{\endproof}
\newenvironment{proofconverse}{%
  \proof}{\endproof}

\begin{document}

\title{Matching of Markov Databases Under Random Column Repetitions\\
\thanks{This work is supported by NYU WIRELESS Industrial Affiliates and National Science Foundation grant CCF-1815821.}}

\begin{comment}
\author{\IEEEauthorblockN{Serhat Bakırtaş}
\IEEEauthorblockA{\textit{Dept. of Electrical and Computer Engineering} \\
\textit{New York University}\\
NY, USA \\
serhat.bakirtas@nyu.edu}
\and 
\IEEEauthorblockN{Elza Erkip}
\IEEEauthorblockA{\textit{Dept. of Electrical and Computer Engineering} \\
\textit{New York University}\\
NY, USA \\
elza@nyu.edu}
}
\end{comment}

\author{Serhat Bakirtas, Elza Erkip\\
 NYU Tandon School of Engineering\\
Emails: \{serhat.bakirtas, elza\}@nyu.edu }

\maketitle

\begin{abstract}
Matching entries of correlated shuffled databases have practical applications ranging from privacy to biology. In this paper, motivated by synchronization errors in the sampling of time-indexed databases, matching of random databases under random column repetitions and deletions is investigated. It is assumed that for each entry (row) in the database, the attributes (columns) are correlated, which is modeled as a Markov process. Column histograms are proposed as a permutation-invariant feature to detect the repetition pattern, whose asymptotic-uniqueness is proved using information-theoretic tools. Repetition detection is then followed by a typicality-based row matching scheme. Considering this overall scheme, sufficient conditions for successful matching of databases in terms of the database growth rate are derived. A modified version of Fano's inequality leads to a tight necessary condition for successful matching, establishing the matching capacity under column repetitions. This capacity is equal to the erasure bound, which assumes the repetition locations are known a-priori. Overall, our results provide insights on privacy-preserving publication of anonymized time-indexed data.

\end{abstract}

\begin{comment}
\begin{IEEEkeywords}
\todo{Some keywords maybe?}
\end{IEEEkeywords}
\end{comment}

\section{Introduction}
\label{sec:introduction}
Recently, with the proliferation of smart devices and the emergence of big data applications, there has been a growing concern over potential privacy leakage from \emph{anonymized} data, approached from legal~\cite{ohm2009broken} and corporate~\cite{bigdata} points of view. These concerns are also articulated in the respective literatures through successful practical de-anonymization attacks on real data~\cite{naini2015you,datta2012provable,narayanan2008robust,sweeney1997weaving,takbiri2018matching,wondracek2010practical,su2017anonymizing,shusterman2019robust,gulmezoglu2017perfweb,bilge2009all,srivatsa2012deanonymizing,cheng2010you,kinsella2011m,kim2016inferring}.

In the light of the above practical attacks, several groups initiated rigorous analyses of the graph matching problem~\cite{erdos1960evolution,babai1980random,janson2011random,czajka2008improved,yartseva2013performance,pedarsani2013bayesian,fiori2013robust,lyzinski2014seeded,onaran2016optimal,cullina2016improved}. Correlated graph matching has applications beyond privacy, such as image processing~\cite{sanfeliu2002graph} and DNA sequencing, which is shown to be equivalent to matching bipartite graphs~\cite{blazewicz2002dna}. Matching of correlated databases, also equivalent to bipartite graph matching, have been investigated from information-theoretic~\cite{shirani8849392,cullina,dai2019database,bakirtas2021database,bakirtas2022seeded} and statistical~\cite{kunisky2022strong} perspectives. In \cite{cullina}, Cullina \emph{et al.} introduced \textit{cycle mutual information} as a correlation metric and derived sufficient conditions for successful matching and a converse result using perfect recovery as error criterion. In \cite{shirani8849392}, Shirani \emph{et al.} considered a pair of databases of the same size, and drawing an analogy between channel decoding and database matching, derived necessary and sufficient conditions on the database growth rate for successful database matching. In~\cite{dai2019database}, Dai \emph{et al.} considered the matching of a pair of databases with jointly Gaussian features with perfect recovery constraint. Similarly, in~\cite{kunisky2022strong}, Kunisky and Niles-Weed considered the same problem from the statistical perspective in different regimes of database size and under several recovery criteria.
\begin{figure}[t]
\centerline{\includegraphics[width=0.5\textwidth,trim={0cm 12cm 2cm 0},clip]{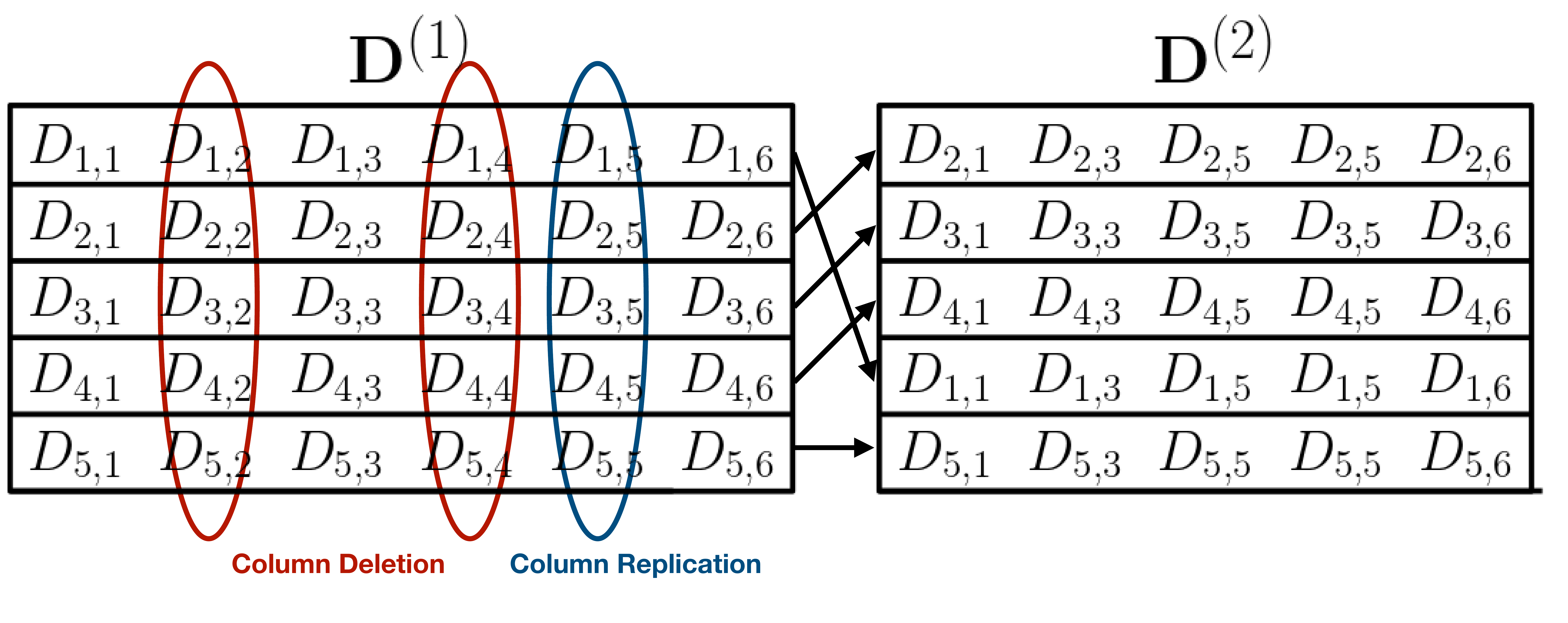}}
\caption{An illustrative example of database matching under column repetitions. The columns circled in red are deleted whereas the column circled in blue is repeated twice, \emph{i.e.,} replicated. Our goal is to estimate the row permutation $\boldsymbol{\Theta}$, which in this example given as; $\boldsymbol{\Theta}(1)=4$, $\boldsymbol{\Theta}(2)=1$, $\boldsymbol{\Theta}(3)=2$,
$\boldsymbol{\Theta}(4)=3$ and $\boldsymbol{\Theta}(5)=5$, by matching the rows of $\mathbf{D}^{(1)}$ and $\mathbf{D}^{(2)}$. Here the $i$\textsuperscript{th} row of $\mathbf{D}^{(1)}$ corresponds to the $\Theta(i)$\textsuperscript{th} row of $\mathbf{D}^{(2)}$.}
\label{fig:intro}
\end{figure}

Motivated by synchronization errors in sampling of time-series datasets, in our prior work we considered database matching under \emph{random column deletions}~\cite{bakirtas2021database}. Assuming an \emph{i.i.d.} underlying distribution for database attributes (columns) and the same successful matching criterion as~\cite{shirani8849392},  we derived an achievable database growth rate assuming a probabilistic side information on the deletion locations. We also proposed an algorithm to extract the side information on the deletion locations from a batch of already-matched rows, called \emph{seeds.}

In this paper, we generalize~\cite{bakirtas2021database}, by assuming a more general model for synchronization errors, namely \emph{column repetitions} where in addition to some columns being deleted as in~\cite{bakirtas2021database}, some columns may be sampled several times consecutively, \emph{i.e.,} \emph{replicated}, as illustrated in Figure~\ref{fig:intro}. Furthermore, in order to account for the potential correlation among the attributes (columns), we model the rows using a Markov process contrary. Under this generalized model, we derive an improved achievable database growth rate. We propose a novel \emph{histogram-based} \emph{repetition detection} algorithm, where we compare the column histograms in order to infer the column repetition pattern with high probability, followed by a typicality-based matching scheme to match the rows of the correlated databases. We also derive the necessary conditions for successful matching. We show that the necessary and sufficient conditions are tight up to equality, and equal to the erasure bound,  which is obtained when there are no replications and the deletion locations are perfectly known. Thus, we completely characterize the \emph{capacity} of the matching of column repeated databases.
 
The organization of this paper is as follows: Section \ref{sec:problemformulation} introduces the problem formulation. In Section \ref{sec:achievability}, our main result on matching capacity and the proof of the achievability are presented. In Section \ref{sec:converse}, the converse is proved. Finally, in Section \ref{sec:conclusion} the results and ongoing work are discussed.

\noindent{\em Notation:} We denote the set of integers $\{1,...,n\}$ as $[n]$, and matrices with uppercase bold letters. For a matrix $\mathbf{D}$, $D_{i,j}$ denotes the $(i,j)$\textsuperscript{th} entry. Furthermore, by $A^n$, we denote a row vector consisting of scalars $A_1,\dots,A_n$ and the indicator of event $\epsilon$ by $\mathbbm{1}_\epsilon$. The logarithms, unless stated explicitly, are in base $2$.

\section{Problem Formulation}
\label{sec:problemformulation}

We use the following definitions, some of which are similar to~\cite{shirani8849392,bakirtas2021database,bakirtas2022seeded} to formalize our problem.

\begin{defn}{\textbf{(Unlabeled Markov Database)}}\label{defn:markovdb}
An ${(m_n,n,\mathbf{P})}$ \emph{unlabeled Markov database} is a randomly generated ${m_n\times n}$ matrix ${\mathbf{D}=\{D_{i,j}\in\mathfrak{X}:i\in[m_n],j\in[n]\}}$ whose rows are \emph{i.i.d.} and follow a first-order stationary Markov process defined over the alphabet ${\mathfrak{X}=\{1,\dots,|\mathfrak{X}|\}}$ with probability transition matrix $\mathbf{P}$ such that
\begin{align}
    \mathbf{P} &= \gamma \mathbf{I} + (1-\gamma) \mathbf{U}\label{eq:markovtransitionmatrix}\\
    U_{i,j} &= u_j>0, \: \forall (i,j)\in \mathfrak{X}^2
\end{align}
and 
\begin{gather}
        \sum\limits_{j\in\mathfrak{X}} u_j =1\\
    \gamma \in \Big(-\min\limits_{j\in\mathfrak{X}}\frac{u_j}{1-u_j},1\Big)
%    \\\pi P = \pi 
\end{gather}
where $\mathbf{I}$ is the identity matrix. It is assumed that ${D_{i,1}\overset{\text{i.i.d.}}{\sim}\pi=[u_1,\dots,u_{|\mathfrak{X}|}]}$, $i=1,\dots,m_n$,  where $\pi$ is the stationary distribution associated with $\mathbf{P}$.
\end{defn}

\begin{defn}{\textbf{(Column Repetition Pattern)}}
The \emph{column repetition pattern} $S^n$ is a random vector consisting of \emph{i.i.d.} elements $S_j$, $j\in [n]$, drawn from a discrete probability distribution $p_S$ with a finite discrete support ${\{0,\dots,s_{\max}\}}$. The parameter ${\delta\triangleq p_S(0)}$ is called the \emph{deletion probability}.
\end{defn}

\begin{defn}{\textbf{(Labeled Repeated Database)}}\label{defn:labeleddb}
Let $\mathbf{D}^{(1)}$ be an ${(m_n,n,P)}$ unlabeled Markov database, $S^n$ be the column repetition pattern, and $\boldsymbol{\Theta}_n$ be a uniform permutation of $[m_n]$ with $\mathbf{D}^{(1)}$, $S^n$ and $\boldsymbol{\Theta}_n$ independently chosen. Given $\mathbf{D}^{(1)}$ and $S^n$, the pair ${(\mathbf{D}^{(2)},\boldsymbol{\Theta}_n)}$ is called the \emph{labeled repeated database} if the $(i,j)$\textsuperscript{th} element $D_{i,j}^{(1)}$ of $\mathbf{D}^{(1)}$ and its counterpart $D_{i,j}^{(2)}$ in $\mathbf{D}^{(2)}$ have the following relation:
\begin{align}
D^{(2)}_{i,j}&=
    \begin{cases}
      E , &  \text{if } S_{j}=0\\
      D^{(1)}_{\boldsymbol{\Theta}_n^{-1}(i),j}\otimes {1}^{S_{j}} & \text{if } S_{i}\ge 1
    \end{cases} 
\end{align}
where ${1}^{S_{j}}$ and $\otimes$ denote the all-ones row vector of length $S_j$ and the Kronecker product, respectively. Furthermore ${D^{(2)}_{i,j}=E}$ corresponds to the empty string and ${D^{(2)}_{i,j}=D^{(1)}_{\boldsymbol{\Theta}_n^{-1}(i),j}\otimes 1^{S_{j}}}$ corresponds to the $j$\textsuperscript{th} column of $\mathbf{D}^{(2)}$ being an $m_n\times S_j$ matrix consisting of $S_j$ copies of the $j$\textsuperscript{th} column of $\mathbf{D}^{(1)}$, concatenated together after shuffling with $\boldsymbol{\Theta}_n$. $\boldsymbol{\Theta}_n$ and $\mathbf{D}^{(2)}$ are called the \emph{labeling function} and \emph{correlated column repeated database}, respectively. The respective rows $D^{(1)}_{i_1}$ and $D^{(2)}_{i_2}$ of $\mathbf{D}^{(1)}$ and $\mathbf{D}^{(2)}$ are said to be \emph{matching rows}, if ${\boldsymbol{\Theta}_n(i_1)=i_2}$.

If ${S_{j}=0}$, the $j$\textsuperscript{th} column of $\mathbf{D}^{(1)}$ is said to be \emph{deleted} and if ${S_{j}>1}$, $j$\textsuperscript{th} column of $\mathbf{D}^{(1)}$ is said to be \emph{replicated}.
\end{defn}

In our model, the correlated column repeated database $\mathbf{D}^{(2)}$ is obtained by permuting the rows of the unlabeled Markov database $\mathbf{D}^{(1)}$ with the uniform permutation $\boldsymbol{\Theta}_n$ followed by column repetition based on the repetition pattern $S^n$. We further assume that there is no noise on the retained entries, as is often done in the repeat channel literature~\cite{cheraghchi2020overview}. 

\begin{defn}{\textbf{(Successful Matching Scheme)}}
A \emph{matching scheme} is a sequence of mappings $\varphi_n: (\mathbf{D}^{(1)},\mathbf{D}^{(2)})\to \hat{\boldsymbol{\Theta}}_n $ where $\mathbf{D}^{(1)}$ is the unlabeled database, $\mathbf{D}^{(2)}$ is the correlated column repeated database and $\hat{\boldsymbol{\Theta}}_n$ is the estimate of the true labeling function $\boldsymbol{\Theta}_n$. The scheme $\varphi_n$ is said to be \emph{successful} if 
\begin{align}
    \lim\limits_{n\to\infty}\Pr\left(\boldsymbol{\Theta}_n(J)\neq\hat{\boldsymbol{\Theta}}_n(J)\right)&\to 0
\end{align}
where $J\sim\text{Unif}([m_n])$. Here the event $\boldsymbol{\Theta}_n(J)\neq\hat{\boldsymbol{\Theta}}_n(J)$ is called the \emph{matching error}.
\end{defn}

Similar to \cite{shirani8849392,bakirtas2021database,bakirtas2022seeded}, our performance metric is the probability of mismatch of a uniformly chosen row. This formulation allows us to derive results for a wide set of database distributions. Note that this performance metric is different than those of~\cite{cullina,dai2019database}, where the probability of the perfect recovery of the complete labeling function $\boldsymbol{\Theta}_n$ was considered.

The relationship between the row size $m_n$ and the column size $n$ of the unlabeled database affects the probability of matching error in the following fashion: For a given $n$, as $m_n$ increases, so does the probability of matching error due to the increased number of candidate rows. As stated in~\cite[Theorem 1.2]{kunisky2022strong}, for the setting in our paper, the regime of interest is $m_n$ growing exponentially in $n$.
\begin{defn}\label{defn:dbgrowthrate}{\textbf{(Database Growth Rate)}}
The \emph{database growth rate} $R$ of an unlabeled Markov database with $m_n$ rows and $n$ columns is defined as 
\begin{align}
    R&=\lim\limits_{n\to\infty} \frac{1}{n}\log m_n
\end{align}
\end{defn}

\begin{defn}{\textbf{(Achievable Database Growth Rate)}}\label{defn:achievable}
Consider a sequence of ${(m_n,n,P)}$ unlabeled Markov databases, a repetition probability distribution $p_S$ and the resulting labeled repeated databases. A database growth rate $R$ is said to be \emph{achievable} if there exists a successful matching scheme when the unlabeled database has growth rate $R$.
\end{defn}
\begin{defn}{\textbf{(Matching Capacity)}}\label{defn:matchingcapacity}
The \emph{matching capacity} $C$ is the supremum of the set of all achievable rates corresponding to a probability transition matrix $\mathbf{P}$ and a repetition probability distribution $p_S$.
\end{defn}

Our goal in this paper is to characterize the matching capacity of the database matching problem under the aforementioned Markovian row process and random column repetition models.
\section{Matching Capacity and Achievability}
\label{sec:achievability}
Theorem~\ref{thm:mainresult} below presents our main result on the matching capacity. We prove the achievability part of Theorem~\ref{thm:mainresult} in this section and the converse in Section~\ref{sec:converse}.

\begin{thm}{\textbf{(Matching Capacity Under Column Repetitions)}}\label{thm:mainresult}
Consider a probability transition matrix $\mathbf{P}$ and a repetition probability distribution $p_S$. Then, the matching capacity is
\begin{align}
    C &= (1-\delta)^2 \sum\limits_{r=0}^\infty \delta^r H(X_0|X_{-r-1})\label{eq:mainresult}
\end{align}
where $\delta\triangleq p_S(0)$ is the deletion probability and $H(X_0|X_{-r-1})$ is the entropy rate associated with the probability transition matrix
\begin{align}
    \mathbf{P}^{r+1}&=\gamma^{r+1} \mathbf{I}+(1- \gamma^{r+1}) \mathbf{U}/ \label{eq:Ppower}
\end{align}
The capacity can further be simplified as
\begin{align}
    C &= \frac{(1-\delta)(1-\gamma)}{(1-\gamma\delta)} [H(\pi)+\sum\limits_{i\in\mathfrak{X}} u_i^2\log u_i] \notag\\
    &\hspace{1em}- (1-\delta)^2 \sum\limits_{r=0}^\infty \delta^r \sum\limits_{i\in \mathfrak{X}} u_i (\gamma^{r+1}+(1-\gamma^{r+1})u_i)\notag\\
    &\hspace{10em}\log (\gamma^{r+1}+(1-\gamma^{r+1})u_i)\label{eq:thm1eval}
\end{align}
where $H(\pi)$ denotes the entropy of the stationary distribution $\pi$.
\end{thm}

Observe that the RHS of~\eqref{eq:mainresult} is the mutual information rate for an erasure channel with erasure probability $\delta$ with first-order Markov $(\mathbf{P})$ inputs, as given in~\cite{li2014input}. Thus, Theorem~\ref{thm:mainresult} states that we can achieve the erasure bound which assumes a-priori knowledge of the column repetition pattern. This is unlike channel synchronization problems, where the erasure bound is a loose upper bound on the capacity. As we see below, the identicality of the repetition pattern across rows allows us to detect repetitions using a collapsed histogram based detection applied to columns. Finally, the matching capacity depends on the repetition distribution only through the deletion probability $\delta=p_S(0)$, rendering the replicated columns of the database irrelevant, as discussed in Section~\ref{sec:converse}. 

Note that the special case where $\gamma=0$ results in an \emph{i.i.d.} database distribution. Thus we have the following corollary:
\begin{cor}{(\textbf{i.i.d. Database Columns})}\label{cor:iid}
When the database entries $D_{i,j}$ are drawn \emph{i.i.d.} from $\mathfrak{X}$ according to $p_X$, the matching capacity becomes
\begin{align}
    C = (1-\delta) H(X)
\end{align}
where $\delta\triangleq p_S(0)$ is the deletion probability.
\end{cor}
Note that Corollary~\ref{cor:iid} improves the achievability result of~\cite{bakirtas2021database}. Therefore, in addition to generalizing~\cite{bakirtas2021database} to Markov databases and column repetitions, this paper also improves the achievability result of~\cite{bakirtas2021database}.

To prove the achievability of the matching capacity in Theorem~\ref{thm:mainresult}, we consider the following two-phase matching scheme: Given $\mathbf{D}^{(1)}$ and $\mathbf{D}^{(2)}$, the unlabeled and the correlated column repeated databases, we first infer the underlying repetition pattern $S^n$ using the \emph{collapsed histogram based detection} algorithm on the column histograms of $\mathbf{D}^{(1)}$ and $\mathbf{D}^{(2)}$. Then, we use a joint typicality based sequence matching scheme to match the rows of $\mathbf{D}^{(1)}$ and $\mathbf{D}^{(2)}$.

The collapsed histogram based repetition detection algorithm works as follows: First, for tractability, we \say{collapse} the Markov chain into a binary-valued one. We pick a symbol $x$ from the alphabet $\mathfrak{X}$, WLOG $x=1$, and define the \emph{collapsed} databases $\tilde{\mathbf{D}}^{(1)}$ and $\tilde{\mathbf{D}}^{(2)}$ as follows:
\begin{align}
    \tilde{\mathbf{D}}^{(r)}_{i,j} &= \begin{cases}
    1 & \text{if } {\mathbf{D}}^{(r)}_{i,j} = 1\\
    2 & \text{if } {\mathbf{D}}^{(r)}_{i,j} \neq 1
    \end{cases}, \: \forall (i,j),\: r=1,2
\end{align}
From~\cite[Theorem 3]{burke1958markovian} and \eqref{eq:markovtransitionmatrix}, the rows of the collapsed database $\tilde{\mathbf{D}}^{(1)}$ become \emph{i.i.d.} first-order stationary binary Markov chains, with the following probability transition matrix and stationary distribution:
\begin{align}
    \tilde{\mathbf{P}}&=\begin{bmatrix}\gamma+(1-\gamma) u_1 & (1-\gamma)(1-u_1)\\
    (1-\gamma)u_1 & 1-(1-\gamma)u_1
    \end{bmatrix}\\
    \tilde{\pi}&=\begin{bmatrix}
    u_1 & 1-u_1
    \end{bmatrix}
\end{align}
Note that after collapsing the Markov chain, the histogram of the $j$\textsuperscript{th} column of $\tilde{\mathbf{D}}^{(1)}$ can be represented by the scalar $\tilde{H}_j^{(1)}$ which denotes the number of occurrences of state 2 in the $j$\textsuperscript{th} column of $\tilde{\mathbf{D}}^{(1)}$, $j\in[n]$. More formally, we have
\begin{align}
    \tilde{H}^{(1)}_j &\triangleq \sum\limits_{i=1}^{m_n} \mathbbm{1}_{\left[\tilde{D}^{(1)}_{i,j}= 2 \right]},\forall j\in [n]\label{eq:histogramdefn}
\end{align}

Our histogram-based detection algorithm exploits two facts: First, the histogram (equivalently the type) of each column of $\tilde{\mathbf{D}}^{(1)}$ and $\tilde{\mathbf{D}}^{(2)}$ is invariant to row permutations. Second, as we prove in Lemma~\ref{lem:histogram}, the histogram of each column is asymptotically-unique due to the row size $m_n$ being exponential in the column size $n$. Finally, since there is no noise on the retained entries of $\mathbf{D}^{(2)}$, we can match the column histograms, present in both $\tilde{\mathbf{D}}^{(1)}$ and $\tilde{\mathbf{D}}^{(2)}$ and detect the deleted columns, in an error-free fashion.

The following lemma provides conditions for the asymptotic-uniqueness of column histograms ${\tilde{H}_j^{(1)}}$, ${j\in[n]}$.

\begin{lem}{\textbf{(Asymptotic Uniqueness of the Column Histograms)}}\label{lem:histogram}
Let $\tilde{H}^{(1)}_j$ denote the histogram of the $j$\textsuperscript{th} column of $\tilde{\mathbf{D}}^{(1)}$, as in~\eqref{eq:histogramdefn}.
Then,
\begin{align}
    \Pr\left(\exists i,j\in [n],\: i\neq j,\tilde{H}^{(1)}_i=\tilde{H}^{(1)}_j\right)\to 0 \text{ as }n\to \infty
\end{align}
if $m_n=\omega(n^4)$.
\end{lem}
\begin{proof}
See Appendix~\ref{proof:histogram}.
\end{proof}
\begin{comment}
We note that even though the order relation given for the \emph{i.i.d.} distribution in~\cite{bakirtas2022database} is $m=\omega(n^{\frac{4}{|\mathfrak{X}|-1}})$, indicating that the larger alphabet sizes decreases the sufficient order relation between $m$ and $n$. On the other hand, the sufficient order relation we derive in this work is independent of the alphabet size. This is due to a slight change in the upper bounding strategy we use in our proof. For the sake of the simplicity, when we compare the histogram of the columns, we only compare the frequency of a single symbol, collapsing the rest to another one, converting the $|\mathfrak{X}|$-ary detection into a binary one. Thus we obtain the order relation of the \emph{i.i.d.} distribution case with an binary virtual alphabet. One can show that the same order relation is indeed sufficient in our Markov database scenario. \todo{I'll mention the numerical experiments agreeing with this argument here and have a figure.} 
\end{comment}
Note that by Definition~\ref{defn:dbgrowthrate}, $m_n$ is exponential in $n$ and the order relation of Lemma~\ref{lem:histogram} is automatically satisfied.

Next, we present the proof of the achievability part of Theorem~\ref{thm:mainresult}.

\begin{proofachievable}
Let $S^n$ be the underlying repetition pattern and $K\triangleq\sum_{i=1}^n S_i$ be the number of columns in $\mathbf{D}^{(2)}$. Our matching scheme consists of the following steps:
\begin{enumerate}[label=\textbf{\arabic*)},leftmargin=1.3\parindent]
    \item Construct the collapsed histogram vectors $\tilde{{H}}^{(1),n}$ and $\tilde{{H}}^{(2),K}$ as 
\begin{align}
    \tilde{H}_j^{(r)}&=\sum\limits_{i=1}^{m_n} \mathbbm{1}_{\left[\tilde{D}^{(r)}_{i,j}=2 \right]},\quad
    \begin{cases}
    \forall j\in [n],&\text{if } r=1  \\
    \forall j\in [K] & \text{if } r=2
    \end{cases}
\end{align}
\item Check the uniqueness of the entries $\tilde{H}^{(1)}_j$ $j\in[n]$ of $\tilde{{H}}^{(1),n}$. If there are at least two which are identical, declare a \emph{detection error} whose probability is denoted by $\mu_n$. Otherwise, proceed with Step~3.
\item If $\tilde{H}^{(1)}_j$ is absent in $\tilde{{H}}^{(2),K}$, declare it deleted, assigning $\hat{S}_j=0$. Note that, conditioned on the uniqueness of the column histograms $\tilde{H}^{(1)}_j$ $\forall j\in[n]$, this step is error free.
\item If $\tilde{H}^{(1)}_j$ is present $s\ge 1$ times in $\tilde{{H}}^{(2),K}$ , assign $\hat{S}_j=s$. Again, if there is no detection error in Step~2, this step is error free. Note that at the end of this step, provided there are no detection errors, we recover $S^n$, \emph{i.e.}, $\hat{{S}}^n={S}^n$.
\item Based on $\hat{{S}}^n$, $\mathbf{D}^{(1)}$ and $\mathbf{D}^{(2)}$, construct $\bar{\mathbf{D}}^{(2)}$ as the following:
\begin{itemize}
    \item If $\hat{S}_j = 0$, the $j$\textsuperscript{th} column of $\bar{\mathbf{D}}^{(2)}$ is a column consisting of erasure symbol $\ast\notin\mathfrak{X}$.
    \item If $\hat{S}_j \ge 1$, the $j$\textsuperscript{th} column of $\bar{\mathbf{D}}^{(2)}$ is the $j$\textsuperscript{th} column of $\mathbf{D}^{(1)}$.
\end{itemize}
Note that after the removal of the additional replicas and the introduction of the erasure symbols, $\bar{\mathbf{D}}^{(2)}$ has $n$ columns.
\item Fix $\epsilon>0$. Let $p_{Y|X}$ be the probability transition matrix of an erasure channel with erasure probability $\delta$, that is
\begin{align}
    p_{Y|X}(y|x) &= \begin{cases}
    1-\delta &\text{if }y=x\\
    \delta &\text{if }y=\varepsilon
    \end{cases},\hspace{1em}\forall(x,y)\in\mathfrak{X}^2 \label{eq:erasure}
\end{align}
We consider the input to the memoryless erasure channel as the $i$\textsuperscript{th} row $X^n_i$ of $\mathbf{D}^{(1)}$. The output $\bar{Y}^n$ is the matching row of $\bar{\mathbf{D}}^{(2)}$. For our row matching algorithm, we match the $l$\textsuperscript{th} row $\bar{{Y}}^n_{l}$ of $\bar{\mathbf{D}}^{(2)}$ with the $i$\textsuperscript{th} row $X^n_i$ of $\mathbf{D}^{(1)}$, if $X^n_i$ is the only row of $\mathbf{D}^{(1)}$ jointly $\epsilon$-typical~\cite[Chapter 3]{cover2006elements} with $\bar{{Y}}^n_l$ with respect to $p_{X^n,Y^n}$, where
\begin{align}
    p_{X^n,Y^n}(x^n,y^n) &= p_{X^n}(x^n) \prod\limits_{j=1}^n p_{Y|X}(y_j|x_j)\label{eq:markovinput}
\end{align}
where $X^n$ denotes the Markov chain of length $n$ with probability transition matrix $\mathbf{P}$. This results in $\hat\Theta(1)=l$. 
Otherwise, declare \emph{collision error}.
\end{enumerate}

Denote the $\epsilon$-typical set of sequences (with respect to $p_{X^n}$) by $A_{\epsilon}^{(n)}(X)$ and the jointly $\epsilon$-typical set of sequences (with respect to $p_{X^n,Y^n}$) by $A_{\epsilon}^{(n)}(X,Y)$. Supposing the true label for the $l$\textsuperscript{th} row $\bar{Y}^n_l$ of $\bar{\mathbf{D}}^{(2)}$ is $1$, \emph{i.e.,} $\boldsymbol{\Theta}_n(1)=l$, and denoting the pairwise collision probability between $X^n_1$ and $X^n_i$, by $P_{col,i}$, for any $i\neq 1$ we have
\begin{align}
    P_{col,i}&=\Pr(({X}^n_i,\bar{{Y}}^n_l)\in A_\epsilon^{(n)})\\
    &\le 2^{-n(I(X;Y)-3\epsilon)}\vspace{-1em}
\end{align}
where 
\begin{align}
    I(X;Y) &= \lim\limits_{n\to\infty}\frac{I(X^n;\bar{Y}^n)}{n}
\end{align}
is the mutual information rate of the joint probability distribution $p_{X^n,Y^n}$. Thus, we can bound the probability of error $P_e$ as
\begin{align}
    P_e &\le \mu_n+\Pr({X}^n_1\notin A_\epsilon^{(n)}(X))+\sum\limits_{i=2}^n P_{col,i}\\
    &\le \mu_n+\epsilon+\sum\limits_{i=2}^n 2^{-n(I(X;Y)-3\epsilon)}\\
    &= \mu_n + \epsilon + 2^{n(R-I(X;Y)+3\epsilon)}
\end{align}

Since $m_n$ is exponential in $n$, by Lemma~\ref{lem:histogram}, ${\mu_n\to0}$ as ${n\to\infty}$. Thus
\begin{align}
    P_e&< 3 \epsilon \text{ as }n\to\infty
\end{align}
if 
$R<I(X;Y)-3\epsilon$. Thus, we can argue that any database growth rate $R$ satisfying
\begin{align}
    R&<I(X;Y)\label{eq:achievable}
\end{align}
is achievable, by taking $\epsilon$ small enough. From~\cite[Corollary II.2]{li2014input} we have
\begin{align}
    I(X;Y)&=(1-\delta)^2 \sum\limits_{r=0}^\infty \delta^r H(X_0|X_{-r-1}) \label{eq:MIrate}
\end{align}
where $H(X_0|X_{-r-1})$ is the entropy rate associated with the probability transition matrix $\mathbf{P}^{r+1}$. Finally, we prove \eqref{eq:Ppower} through induction.
By Definition~\ref{defn:markovdb}, \eqref{eq:Ppower} is satisfied for $r=0$. Now assume \eqref{eq:Ppower} is true for some $r\in\mathbb{N}$. In other words,
\begin{align}
    \mathbf{P}^{r}&=\gamma^{r} \mathbf{I}+(1- \gamma^{r}) \mathbf{U}
\end{align}
Observing $\mathbf{U}^k=\mathbf{U}$, $\forall k\in\mathbb{N}$, we obtain
\vspace{-1em}
\begin{adjustwidth}{-0.28cm}{0pt}
\begin{align}
    \mathbf{P}^{r+1}&=(\gamma \mathbf{I}+(1- \gamma) \mathbf{U})(\gamma^{r} \mathbf{I}+(1- \gamma^{r}) \mathbf{U})\\
    &= \gamma^{r+1} \mathbf{I} + (\gamma(1-\gamma^r)+(1-\gamma)\gamma^r) \mathbf{U}\notag\\
    &\hspace{3em}+ (1-\gamma)(1-\gamma^r) \mathbf{U}^2\\
    &=\gamma^{r+1} \mathbf{I} + (\gamma(1-\gamma^r)+(1-\gamma)\gamma^r+(1-\gamma)(1-\gamma^r)) \mathbf{U} \\
    &= \gamma^{r+1} \mathbf{I} + (1-\gamma^{r+1}) \mathbf{U}\label{eq:prplus1}
\end{align}
\end{adjustwidth}
From~\eqref{eq:MIrate}-\eqref{eq:prplus1} and \cite[Theorem 4.2.4]{cover2006elements} we obtain~\eqref{eq:thm1eval},
concluding the achievability part of the proof.
\end{proofachievable}

\section{Converse}
\label{sec:converse}
Theorem~\ref{thm:mainresult} states that we can convert repetitions to erasures, achieving the erasure bound. In this section, we show that the lower bound on the matching capacity $C$ given in Section~\ref{sec:achievability} is in fact tight, by proving it to also be an upper bound on the matching capacity $C$.
\begin{proofconverse}
Here we prove that the erasure bound given in~\eqref{eq:mainresult} is an upper bound on all achievable database growth rates. We adopt a genie-aided proof where the repetition pattern $S^n$ is available a-priori. Furthermore, we use the modified Fano's inequality presented in~\cite{shirani8849392}.

Let $R$ be the database growth rate and $P_e$ be the probability that the scheme is unsuccessful for a uniformly-selected row pair. More formally,
\begin{align}
   P_e&\triangleq \Pr\left(\boldsymbol{\Theta}_n(J)\neq\hat{\boldsymbol{\Theta}}_n(J)\right),\hspace{1em} J\sim\text{Unif}([m_n])
\end{align}
Furthermore, let $S^n$ be the repetition pattern and $K=\sum_{j=1}^n S_j$.

Since $\boldsymbol{\Theta}_n$ is a uniform permutation, from Fano's inequality, we have
\begin{align}
    H(\boldsymbol{\Theta}_n|\mathbf{D}^{(1)},\mathbf{D}^{(2)})
    &\le 1+P_e \log(m_n!)\\
    &\le 1+P_e m_n \log m_n\label{eqn:mfaclessthanmm}
\end{align}
where \eqref{eqn:mfaclessthanmm} follows from $m_n!\le m_n^{m_n}$. Thus, we get
\begin{align}
    H(\boldsymbol{\Theta}_n)&=H(\boldsymbol{\Theta}_n|\mathbf{D}^{(1)},\mathbf{D}^{(2)})+I(\boldsymbol{\Theta}_n;\mathbf{D}^{(1)},\mathbf{D}^{(2)})\\
    &\le 1+P_e m_n \log m_n+I(\boldsymbol{\Theta}_n;\mathbf{D}^{(1)},\mathbf{D}^{(2)})\label{eq:conversefinal1}
\end{align}
Note that
\begin{align}
    I(\boldsymbol{\Theta}_n;\mathbf{D}^{(1)},\mathbf{D}^{(2)})
    &=I(\boldsymbol{\Theta}_n;\mathbf{D}^{(2)})+I(\boldsymbol{\Theta}_n;\mathbf{D}^{(1)}|\mathbf{D}^{(2)})\\
    &= I(\boldsymbol{\Theta}_n;\mathbf{D}^{(1)}|\mathbf{D}^{(2)})\label{eq:converseMI1}\\
    &\le I(\boldsymbol{\Theta}_n,\mathbf{D}^{(2)};\mathbf{D}^{(1)})
\end{align}
where in \eqref{eq:converseMI1} we have used the independence of $\boldsymbol{\Theta}_n$ and $\mathbf{D}^{(2)}$.
\begin{align}
I(\boldsymbol{\Theta}_n,\mathbf{D}^{(2)};\mathbf{D}^{(1)})
    &\le I(\boldsymbol{\Theta}_n,\mathbf{D}^{(2)},\mathbf{S}^n;\mathbf{D}^{(1)})\\
    &= I(\mathbf{D}^{(1)};\boldsymbol{\Theta}_n,\mathbf{D}^{(2)}|S^n)\label{eq:CSindependent}\\
    &= \sum\limits_{i=1}^{m_n} I(D^{(1),n}_i;D^{(2),K}_{\Theta_n^{-1}(i)}|S^n)\\
    &= m_n I(D^{(1),n}_1;D^{(2),K}_{\Theta_n^{-1}(1)}|S^n)\\
    &= m_n I(D^{(1),n}_1;D^{(2),K}_{\Theta_n^{-1}(1)},S^n)\label{eq:rowsidenticallydistributed}
\end{align}
where \eqref{eq:CSindependent}-\eqref{eq:rowsidenticallydistributed} follow from the fact that $\mathbf{D}^{(1)}$ and $S^n$ are independent and non-matching rows are \emph{i.i.d.} conditioned on the repetition pattern $S^n$. 

Now, for brevity let ${X^n=D^{(1),n}_1}$, ${Y^K=D^{(2),K}_{\Theta_n^{-1}(1)}}$ and 
$\tilde{Y}^n$ be obtained from $Y^K$ as described in Step~5 of the achievability proof. Since there is a bijective mapping between $(Y^K,S^n)$ and $(\bar{{Y}}^n,{S}^n)$, we have
\begin{align}
    I({X}^n;{Y}^K,{S}^n) &= I({X}^n;\bar{{Y}}^n,{S}^n)\label{eq:noadditionalinfo1}\\
    &= I({X}^n;\bar{{Y}}^n) + I({X}^n;{S}^n|\bar{{Y}}^n)\\
    &= I({X}^n;\bar{{Y}}^n) \label{eq:noadditionalinfo2}
\end{align}
where \eqref{eq:noadditionalinfo2} follows from the independence of ${S}^n$ from ${X}^n$ conditioned on $\bar{{Y}}^n$. This is because since $\bar{{Y}}^n$ is stripped of all extra replicas, from $(X^n,\bar{{Y}}^n)$ we can only infer the zeros of $S^n$, which is already known through $\bar{{Y}}^n$ via erasure symbols.

Finally, from Stirling's approximation and the uniformity of $\boldsymbol{\Theta}_n$, we have
\begin{align}
\lim\limits_{n\to\infty}\frac{1}{m_n n}H(\boldsymbol{\Theta}_n)&= \lim\limits_{n\to\infty}\Big[ \frac{1}{n}\log m_n + \frac{1}{m_n n}O(n)\Big] = R\label{eq:limHR}
\end{align}

Therefore, from~\eqref{eq:conversefinal1}-\eqref{eq:limHR}, we have
\begin{align}
    \lim\limits_{n\to\infty}\frac{1}{m_n n}H(\boldsymbol{\Theta}_n)&\le \lim\limits_{n\to\infty}\left[ \frac{1}{m_n n}+P_e\frac{1}{n}\log m_n+I({X}^n;\bar{{Y}}^n)\right]\\
    R&\le \lim\limits_{n\to\infty}\frac{I({X}^n;\bar{{Y}}^n)}{n}\label{eqn:converselast}\\
    &= (1-\delta)^2 \sum\limits_{r=0}^\infty \delta^r H(X_0|X_{-r-1}) \label{eq:converseeval}
\end{align}
where \eqref{eqn:converselast} follows from the fact that $P_e\to 0$ as $n\to\infty$ and \eqref{eq:converseeval} follows from~\cite[Corollary II.2]{li2014input}. The rest of the proof follows from the evaluation of~\eqref{eq:converseeval} we did in Section~\ref{sec:achievability}.
\end{proofconverse}

Equations \eqref{eq:noadditionalinfo1}-\eqref{eq:noadditionalinfo2} suggest that the additional copies of the replicated columns do not offer any information. As a result, discarding the additional replicas in the matching scheme of Section~\ref{sec:achievability} does not impact optimality.

\section{Conclusion}
\label{sec:conclusion}
In this paper, we have studied the matching of Markov databases under random column repetitions. By proving the asymptotic-uniqueness of the column histograms of the databases, we have showed that these histograms can be used for the detection of the deleted and replicated columns. Using the proposed histogram-based detection and typicality-based row matching, we have derived an achievability result for database growth rate, which we have showed is tight, thus giving us the database matching capacity. Our ongoing work includes investigating the matching capacity in the presence of noise as well as synchronization errors~\cite{bakirtas2022seeded} and when different subsets of rows are sampled separately and then merged together, \emph{i.e.,} different subsets of rows experience different repetition patterns.

\bibliography{references}
\bibliographystyle{IEEEtran}

\appendix
\subsection{Proof of Lemma~\ref{lem:histogram}}\label{proof:histogram}
We let ${\mu_n\triangleq \Pr(\exists i,j\in [n],\: i\neq j,\tilde{{H}}^{(1)}_{i}=\tilde{{H}}^{(1)}_j)}$. From the union bound we obtain
\begin{align}
    \mu_n&\le \sum\limits_{{(i,j)\in[n]^2:i<j}} \Pr(\tilde{{H}}^{(1)}_{i}=\tilde{{H}}^{(1)}_{j})\\
    &\le n^2 \max\limits_{{(i,j)\in[n]^2:i<j}} \Pr(\tilde{{H}}^{(1)}_{i}=\tilde{{H}}^{(1)}_{j})
\end{align}
Due to stationarity, the maximum is equal to $\Pr(\tilde{{H}}^{(1)}_1=\tilde{{H}}^{(1)}_{s+1})$ for some $s$. For brevity, let $\mathbf{Q}\triangleq\tilde{\mathbf{P}}^s$ and $q\triangleq \Pr(\tilde{{H}}^{(1)}_1=\tilde{{H}}^{(1)}_{s+1})$. Observe that 
$\tilde{{H}}^{(1)}_1$ and $\tilde{{H}}^{(1)}_{s+1}$ are correlated Binom($m_n,1-u_1$) random variables and for any $s$, $\mathbf{Q}$ has positive values, \emph{i.e.,} the collapsed Markov chain is irreducible for any $s$. Now, we have
\vspace{-1em}
\begin{adjustwidth}{-0.25cm}{0pt}
\begin{align}
     q&= \sum\limits_{r=0}^{m_n} \Pr(\tilde{{H}}^{(1)}_1=r) \Pr(\tilde{{H}}^{(1)}_{s+1}=r|\tilde{{H}}^{(1)}_1=r)\\
    &= \sum\limits_{r=0}^{m_n} \binom{m}{r}(1-u_1)^r u_1^{m_n-r} \Pr(\tilde{{H}}^{(1)}_{s+1}=r|\tilde{{H}}^{(1)}_1=r)
\end{align}
\end{adjustwidth}
Note that since the rows of $\tilde{\mathbf{D}}^{(1)}$ are \emph{i.i.d.}, we have 
\begin{align}
    \Pr(\tilde{{H}}^{(1)}_{s+1}=r|\tilde{{H}}^{(1)}_1=r) = \Pr(A+B=r)
\end{align}
where $A\sim\text{Binom}(r,Q_{1,1})$ and $B\sim\text{Binom}(m_n-r,Q_{2,1})$ are independent. Then, from Stirling's approximation and~\cite[Theorem 11.1.2]{cover2006elements}, we get
\begin{align}
    q
    &= \sum\limits_{r=0}^{m_n} \binom{{m_n}}{r}(1-u_1)^r u_1^{{m_n}-r} \Pr(A+B=r)\\
    &\le \frac{e}{\sqrt{2\pi}} {m_n}^{-1/2} \sum\limits_{r=0}^{m_n} \Pi_r^{-1} 2^{-{m_n} D(\frac{r}{{m_n}}\|(1-u_1))}\Pr(A+B=r) 
\end{align}
where $\Pi_r=\frac{r}{{m_n}}(1-\frac{r}{{m_n}})$. Let
\begin{align}
    T &= \sum\limits_{r=0}^{m_n} \Pi_r^{-1} 2^{-{m_n} D(\frac{r}{{m_n}}\|(1-u_1))}\Pr(A+B=r) = T_1+T_2
\end{align}
where
\begin{align}
    T_1 &= \sum_{\mathclap{\hspace{4em} r:D(\frac{r}{{m_n}}\|1-u_1)> \frac{\epsilon_n^2}{2\log_e 2}}} \hspace{2em}\Pi_r^{-1} 2^{-{m_n} D(\frac{r}{{m_n}}\|(1-u_1))}\Pr(A+B=r)\label{eq:T1}\\
    T_2 &= \sum_{\mathclap{\hspace{4em} r:D(\frac{r}{{m_n}}\|1-u_1)\le \frac{\epsilon_n^2}{2\log_e 2}}} \hspace{2em} \Pi_r^{-1} 2^{-{m_n} D(\frac{r}{{m_n}}\|(1-u_1))}\Pr(A+B=r),\label{eq:T2}
\end{align}
$D(\frac{r}{{m_n}}\|(1-u_1))$ denotes the Kullback-Leibler divergence between Bernoulli($\frac{r}{{m_n}}$) and Bernoulli($1-u_1$) distributions, and $\epsilon_n>0$, which is described below in more detail, is such that $\epsilon_n\to0$ as $n\to\infty$.

First, we look at $T_1$. Note that for any $r\in\mathbb{N}$, we have ${\Pi_r\le {m_n}^2}$, suggesting the multiplicative term in the summation in~\eqref{eq:T1} is polynomial with ${m_n}$. Note that we can simply separate the cases $r=0$, $r={m_n}$ whose probabilities vanish exponentially in ${m_n}$. Therefore, as long as ${{m_n} \epsilon_n^2\to\infty}$, $T_1$ has a polynomial number of elements which decay exponentially with ${m_n}$. Thus
\begin{align}
    T_1\to0\text{ as }n\to\infty\label{eq:t1}
\end{align}
as long as ${{m_n} \epsilon_n^2\to\infty}$.

Now, we focus on $T_2$. From Pinsker's inequality~\cite[Lemma 11.6.1]{cover2006elements}, we have
\begin{align}
    D\left(\frac{r}{{m_n}}\Big\|1-u_1\right)\le \frac{\epsilon_n^2}{2\log_e 2}\Rightarrow \text{TV}\left(\frac{r}{{m_n}}, 1-u_1\right)\le \epsilon_n
\end{align}
where TV denotes the total variation distance between the Bernoulli distributions with given parameters. Therefore
\begin{align}
    \Big|\{r:D\Big(\frac{r}{{m_n}}\Big\|1-u_1\Big)\le &\frac{\epsilon_n^2}{2\log_e 2}\}\Big|\\
    &\le \Big|\{r:\text{TV}\Big(\frac{r}{{m_n}}, 1-u_1\Big)\le \epsilon_n\}\Big| \\
    &= O({m_n}\epsilon_n)
\end{align}
for small $\epsilon_n$. Furthermore, when $\text{TV}\left(\frac{r}{{m_n}}, 1-u_1\right)\le \epsilon_n$, we have 
\begin{align}
    \Pi_{r}^{-1} &\le \frac{1}{(1-u_1) u_1}
\end{align}

Now, we investigate $\Pr(A+B=r)$ for the values of $r$ in the interval $[{{m_n}(1-u_1-\epsilon_n)},{{m_n}(1-u_1+\epsilon_n)}]$.
\vspace{-1em}
\begin{adjustwidth}{-0.25cm}{0pt}
\begin{align}
    \Pr(A+B=r)&=\sum\limits_{i=1}^r \Pr(A=r-i) \Pr(B=i)\notag\\[-0.75em]
    &\hspace{4em}+ \Pr(A=r)\Pr(B=0)\\
    &=Q_{1,1}^r Q_{2,2}^{{m_n}-r}+ \sum\limits_{i=1}^r \binom{r}{i} Q_{1,1}^{r-i} (1-Q_{1,1})^i\notag\\
    &\hspace{5em}\binom{{m_n}-r}{i} Q_{2,1}^i (1-Q_{2,1})^{{m_n}-r-i}\label{eq:binom2}
\end{align}
\end{adjustwidth}
Again, from Stirling's approximation on the binomial coefficient in \eqref{eq:binom2} and~\cite[Theorem 11.1.2]{cover2006elements}, we have
\begin{align}
    \Pr(A+B=r)&\le Q_{1,1}^r Q_{2,2}^{{m_n}-r} + \frac{e^2}{2\pi}r^{-1/2}({m_n}-r)^{-1/2} U
\end{align}
where 
\begin{align}
    U = \sum\limits_{i=1}^r \Pi^{-1}_{i/r} \Pi^{-1}_{i/{m_n}-r} 2^{-r D(1-\frac{i}{r}\|Q_{1,1}) -({m_n}-r) D(\frac{i}{{m_n}-r}\|Q_{2,1})}
\end{align}
Then, from $r\in[{{m_n}(1-u_1-\epsilon_n)},{{m_n}(1-u_1+\epsilon_n)}]$ we obtain
\begin{align}
    \Pr(A+B=r)&\le Q_{1,1}^r Q_{2,2}^{{m_n}-r} + \frac{e^2}{2\pi}\frac{{m_n}^{-1}}{\sqrt{(1-u_1-\epsilon_n)(u_1-\epsilon_n)}}  U
\end{align}
and
\begin{align}
    U&\le \sum\limits_{i=1}^r \Pi^{-1}_{i/r} \Pi^{-1}_{i/{m_n}-r}\notag\\
    &\hspace{5em}2^{-{m_n}\left[ (1-u_1-\epsilon_n) D(1-\frac{i}{r}\|Q_{1,1})+(u_1-\epsilon_n) D(\frac{i}{{m_n}-r}\|Q_{2,1})\right]}\\
    &= \sum\limits_{i\notin \mathcal{R}(\epsilon_n)} \Pi^{-1}_{i/r} \Pi^{-1}_{i/{m_n}-r}\notag\\
    &\hspace{5em} 2^{-{m_n}\left[ (1-u_1-\epsilon_n) D(1-\frac{i}{r}\|Q_{1,1})+(u_1-\epsilon_n) D(\frac{i}{{m_n}-r}\|Q_{2,1})\right]}\notag\\
    &+ \sum\limits_{i\in \mathcal{R}(\epsilon_n)} \Pi^{-1}_{i/r} \Pi^{-1}_{i/{m_n}-r} \notag\\
    &\hspace{5em}2^{-{m_n}\left[ (1-u_1-\epsilon_n) D(1-\frac{i}{r}\|Q_{1,1})+(u_1-\epsilon_n) D(\frac{i}{{m_n}-r}\|Q_{2,1})\right]}\label{eq:U}
\end{align}
where we define the set $\mathcal{R}(\epsilon_n)$ as
\begin{align}
    \mathcal{R}(\epsilon_n) &\triangleq\Big\{i\in[r]:D\Big(1-\frac{i}{r}\Big\|Q_{1,1}\Big),D\Big(\frac{i}{{m_n}-r}\Big\|Q_{2,1}\Big)\le \frac{\epsilon_n^2}{2\log_e 2}\Big\}
\end{align}
Note that similar to $T_1$, the first summation in \eqref{eq:U} vanishes exponentially in ${m_n}$ whenever ${m_n}\epsilon_n^2\to\infty$, and using Pinsker's inequality once more, the second term can be upper bounded by
\begin{align}
    O(|\mathcal{R}(\epsilon_n)|)=O({m_n}\epsilon_n)
\end{align}
Now, we choose ${\epsilon_n={m_n}^{-\frac{1}{2}} V_n}$ for some $V_n$ satisfying ${V_n=\omega(1)}$ and ${V_n=o(m_n^{1/2})}$. Thus, $T_1$ vanishes exponentially fast since ${m_n\epsilon_n^2=V_n^2\to\infty}$ and 
\begin{gather}
\Pr(A+B=r) = O(\epsilon_n)\\
    T  = O({m_n} \epsilon_n^2)=O(V_n^2)\\
    \mu_n = O(n^2 {m_n}^{-1/2} V_n^2)
\end{gather}
By the assumption ${m_n=\omega(n^4)}$, we have ${m_n=n^4 W_n}$ for some $W_n$ satisfying  ${\lim\limits_{n\to\infty} W_n=\infty}$. Now, taking ${V_n=o(W_n^{1/4})}$ (e.g. ${V_n=W_n^{1/6}}$), we get
\begin{align}
    \mu_n&\le O( W_n^{-1/2} V_n^2)
    = o(1)
\end{align}
Thus $m_n=\omega(n^4)$ is enough to have $\mu_n\to0$ as $n\to\infty$.\qed
\end{document}